\documentclass[11pt]{amsart}
\pdfoutput=1
\usepackage{times}
\usepackage{comment}
\usepackage{epsfig}
\usepackage{varioref}
\usepackage{textcomp}
\usepackage{multirow}
\usepackage{graphics}
\usepackage{bmpsize}
\usepackage{url}
\usepackage{draftcopy}
\usepackage[numbers]{natbib}
\usepackage{amssymb}
\usepackage{algorithm}
\usepackage{algpseudocode}
\usepackage{varioref}
\usepackage{multirow}
\usepackage{enumitem}
\makeatletter\@ifclassloaded{beamer}{}{\usepackage{imakeidx}}\makeatother
\usepackage{lscape}
\usepackage{gensymb}
\usepackage[toc]{appendix}
\usepackage{mathtools}
\usepackage{pict2e}
\usepackage{picture}

\DeclareMathOperator*{\argmin}{arg~min}

\makeatletter\@ifclassloaded{beamer}{}{}\makeatother

\newcommand{\genlab}[2]{\label{#1:#2}}
\newcommand{\genref}[2]{#1~\vref{#1:#2}}

\newcommand{\Genref}[3]{#1~\vref{#2:#3}}

\newcommand{\seclab}[1]{\genlab{section}{#1}}
\newcommand{\secref}[1]{\genref{section}{#1}}

\newcommand{\eqnlab}[1]{\genlab{equation}{#1}}
\newcommand{\eqnref}[1]{\genref{equation}{#1}}

\newcommand{\Eqnref}[1]{\Genref{Equation}{equation}{#1}}

\newcommand{\figlab}[1]{\genlab{figure}{#1}}
\newcommand{\figref}[1]{\genref{figure}{#1}}

\makeatletter\@ifclassloaded{beamer}{}{%
    \newtheorem{theorem}{Theorem}
    \newtheorem{definition}{Definition}

}
\makeatother
\newcommand{\thelab}[1]{\genlab{theorem}{#1}}
\newcommand{\theref}[1]{\genref{theorem}{#1}}

\newcommand{\Theref}[1]{\Genref{Theorem}{theorem}{#1}}

\newtheorem{lemma}{Lemma}
\newcommand{\lemlab}[1]{\genlab{lemma}{#1}}
\newcommand{\lemref}[1]{\genref{lemma}{#1}}

\newcommand{\Lemref}[1]{\Genref{Lemma}{lemma}{#1}}

\newtheorem{corollary}{Corollary}
\newcommand{\corlab}[1]{\genlab{corollary}{#1}}
\newcommand{\corref}[1]{\genref{corollary}{#1}}

\newcommand{\Corref}[1]{\Genref{Corollary}{corollary}{#1}}

\newcommand{\pdffigure}[4]{\begin{center}\begin{figure}[hbt]\includegraphics[angle=#4,scale=#3]{#1.pdf}\caption[\hspace{\normalparindent}#2]{#2}\figlab{#1}\end{figure}\end{center}}
\makeatletter\@ifclassloaded{beamer}{\input{beamer}}{}\makeatother

\vrefwarning
\begin{document}
%
% metadata
%
\title[Isotropic Correlation Models]{Isotropic Correlation Models\\ for the Cross-Section of Equity Returns}
\author{Graham L. Giller}
\email{graham@gillerinvestments.com}
\date{\today}
\begin{abstract}
This note discusses some of the aspects of a model for the covariance of equity returns based on a simple ``isotropic'' structure in which all pairwise correlations are taken to be the same value. The effect of the structure on feasible values for the common correlation of returns and on the ``effective degrees of freedom'' within the equity cross-section are discussed, as well as the impact of this constraint on the asymptotic Normality of portfolio returns. An eigendecomposition of the covariance matrix is presented and used to partition variance into that from a common ``market'' factor and ``non-diversifiable'' idiosyncratic risk. A empirical analysis of the recent history of the returns of S\&P~500 Index members is presented and compared to the expectations from both this model and linear factor models. This analysis supports the isotropic covariance model and does not seem to provide evidence in support of linear factor models. Analysis of portfolio selection under isotropic correlation is presented using mean-variance optimization for both heteroskedastic and homoskedastic cases. Portfolio selection for negative exponential utility maximizers is also discussed for the general case of distributions of returns with elliptical symmetry. The fact that idiosyncratic risk may not be removed by diversification in a model that the data supports undermines the basic premises of structures such as the C.A.P.M.\ and A.P.T. If the cross-section of equity returns is more accurately described by this structure then an inevitable consequence is that picking stocks is not a ``pointless'' activity, as the returns to residual risk would be non-zero.
\end{abstract}
\maketitle
\section{Isotropic Returns}\seclab{toymodel}
It is obvious to even the most casual observers of equity markets that most assets go up and down together but, not exactly, in the same way on every day. This is exhibited in the dispersion in the sample means of returns and the dispersion in the sample correlations of asset pairs. An enormous amount of research has been executed to address the linked problems of:
\begin{enumerate}[label=(\roman*)]
    \item how to accurately measure the actual covariance of returns; and,
    \item the impact of ``dimensional reduction'' in models for the cross-section of returns and how that effects equilibrium returns if all investors are fully informed about the nature of the equity cross-section.
\end{enumerate}
This issues were first raised by Markowitz\cite{harry1952markowitz}, of course, and addressed in the subsequent works by Sharpe\cite{sharpe1964capm} \textit{et al.} and Ross\cite{ross2013arbitrage}. They are now standard items in the Finance curriculum\footnote{See, for example, Bodie \& Merton \textit{Finance}\cite[pt. IV]{bodie2000finance} and others.} and well known to readers.
\subsection{Definition}
Due to the difficulties associated with intuitively understanding the impact of complex multi-factor models for the equity cross-section, it is appealing to introduce a simple \textit{isotropic} structure in which all pair-wise correlations are equal, and to exhibit results from that ``toy'' model more as an aide to exposition than a serious proposal to model the cross-section of returns. This model is introduced and discussed by Grinold and Kahn in their popular work \textit{Active Portfolio Management}\cite[p. 48]{grinold2000active}.

For assets $i\in[1,N]$ we define the covariance of returns\footnote{The notations $\mathbb{E}[x]$, $\mathbb{V}[x]$ and $\mathbb{V}[x,y]$ are used to mean ``the mean of $x$,'' ``the variance of $x$'' and ``the covariance of $x$ and $y$,'' respectively. $\mathbb{V}[\boldsymbol{x}]$ is used to mean ``the covariance matrix of $\boldsymbol{x}$.} as
\begin{align}\eqnlab{Vdef}
\mathbb{V}[r_{it},r_{jt}]&=\sigma_{i}\sigma_{j}\rho \Leftrightarrow 
\mathbb{V}[\boldsymbol{r}_t]=S_NG_NS_N\ \\
\mathrm{where}\ G_N&=\begin{pmatrix}
1&\rho&\cdots&\rho\\
\rho&1&\cdots&\rho\\
\vdots&&\ddots&\vdots\\
\rho&\rho&\cdots&1
\end{pmatrix}\\ \mathrm{and}
\ S_N&=\begin{pmatrix}
\sigma_1&0&\cdots&0\\
0&\sigma_2&\cdots&0\\
\vdots&&\ddots&\vdots\\
0&0&\cdots&\sigma_N
\end{pmatrix}.
\end{align}

This model is defined for a single time, $t$, and so does not make \textit{any} statements about the stationarity, or lack thereof, of either the asset volatilities, $\sigma_i$, or the correlation parameter $\rho$. For completeness, these parameters may all be read with an implicit time index, i.e.\ $\sigma_{it}$ and $\rho_t$ etc., without any affect on the arguments to follow. In particular, this facilitates the sort of multivariate \textsc{GARCH} models with dynamic correlation coefficients proposed by Engle and Sheppard\cite{engle2001theoretical}. All of the analysis presented here may be considered to represent the returns for some single period $(s,t]$ during which a portfolio is formed ``immediately before'' time $s$ and held, constant, through to the time $t$. The return $r_{it}$ represents the return of asset $i$ over the reference period and the variance, $\sigma_{it|s}^2$, represents the expected variance of those returns given the information set available immediately prior to time $s$.
\subsection{The Variance of Equal-Weighted Portfolio Returns}
Consider a portfolio formed from assets with some general covariance matrix $\Sigma$. If $\boldsymbol{h}_t$ represents the portfolio holdings, by asset, formed ``immediately before'' time $s$ and held to time $t$ then the variance of the returns of that entire portfolio is
\begin{equation}
    V_P=\boldsymbol{h}_t^T\Sigma\boldsymbol{h}_t.
\end{equation}
For an equal-weighted portfolio, with $\boldsymbol{h}_t=\boldsymbol{1}_N/N$ where $\boldsymbol{1}_N$ represents a vector of dimension $N$ where all elements are ones, then it is well known that
\begin{equation}
    V_P=\frac{1}{N^2}\sum_{ij=1}^N\sigma_{ij}\ \mathrm{where}\ \sigma_{ij}=\sigma_i\sigma_{j}\rho_{ij}.
\end{equation}
The $\{\sigma_{ij}\}$ are the elements of the matrix $\Sigma$. It is sometimes convenient to express this quantity in terms of what is called the ``grand sum'' of the matrix, $\Sigma$, which I will write as ``$\mathop{\mathrm{gs}}\Sigma$,'' meaning the sum over all of the elements of the matrix. With this notation, $V_P=(\mathop{\mathrm{gs}}\Sigma)/N^2$ which is equal to the arithmetic mean of all of the elements of $\Sigma$.

This expression may be further decomposed into the sum down the diagonal and twice\footnote{Since $\Sigma$ is a symmetric matrix \textit{by definition}.} the sum of the upper-triangle of $\Sigma$:
\begin{equation}\eqnlab{portvar}
    V_P=\frac{1}{N^2}\sum_{i=1}^N\sigma_i^2+\frac{2}{N^2}\sum_{i=1}^N\sum_{j=i+1}^N\sigma_i\sigma_j\rho_{ij}=V_I+V_C,
\end{equation}
where $V_I$ is the total variance due to the \textit{independent} returns and $V_C$ is the total variance due to their covariance. In matrix notation
\begin{equation}
    V_I=\frac{\mathop{\mathrm{tr}}\Sigma}{N^2}\;\mathrm{and}\;V_C=\frac{\mathop{\mathrm{gs}}\Sigma-\mathop{\mathrm{tr}}\Sigma}{N^2},
\end{equation}
where $\mathop{\mathrm{tr}}\Sigma$ is the usual notation for the trace of the matrix.
\subsection{Independent Returns}
In the special case of independent returns, where $\rho_{ij}=0\ \forall\ i\ne j$, this becomes
\begin{equation}\eqnlab{indepr}
    V_P=V_I=\frac{\overline{\sigma^2}}{N}\ \mathrm{where}\ \overline{\sigma^2}=\frac{1}{N}\sum_{i=1}^N\sigma_i^2
\end{equation}
is the mean variance of the assets in the portfolio. Specializing further to homoskedastic assets,\footnote{In this article the terms ``homoskedastic'' and ``heteroskedastic'' should be interpreted as referring strictly to the cross-section of returns. The analysis presented permits variance to vary through time.} i.e. $\sigma_i=\sigma\ \forall\ i$, we have
\begin{equation}\eqnlab{varmeanr}
    V_P=\frac{\sigma^2}{N}=\mathbb{V}[\overline{r_t}].
\end{equation}
which is well known to be the expression for the error-in-the-mean of $N$ independent random variables and is familiar from the \textit{Law of Large Numbers} as discussed in the context of sampling theory in all elementary texts on Statistics.\footnote{See, for example, Kendall\cite[pp. 308--310]{kendall1999vol1}.} 

Within that context, in the expression for the variance of the sample mean of a statistic $x$,
\begin{equation}\eqnlab{samplevar}
    \mathbb{V}[\overline{x}]=\frac{\mathbb{V}[x]}{N},
\end{equation}
the term $N$ would be referred to as the \textit{degrees of freedom} within the statistic, and represents the ``amount of randomness'' embedded within the quantity under study.
\subsection{The \textsl{Effective} Degrees of Freedom in a Portfolio}
Returning to the more general expression of \eqnref{portvar}, we may write
\begin{equation}
    V_P=\frac{\overline{\sigma^2}}{N}\left(1+\frac{V_C}{V_I}\right).
\end{equation}
Since $V_P$ is strictly non-negative we must have $1+V_C/V_I\ge0$, or $V_C\ge-V_I$, with the lower limit representing the case where all variance is removed due to perfect hedging. Thus this term takes the role of a scale factor correcting the actual sample size, $N$, to an ``effective'' sample size, which is strictly non-negative.

That is, the portfolio variance can be though of as given by
\begin{equation}
    V_P=\frac{\overline{\sigma^2}}{N^*}
\end{equation}
with 
\begin{equation}\eqnlab{Nstar}
    N^*=\frac{N}{1+{V_C}/{V_I}}=N\frac{V_I}{V_P}
\end{equation}
representing the \textit{effective} degrees of freedom within the portfolio and $1/(1+V_C/V_I)$ a degrees of freedom ``correction'' due to the existence of covariance between the asset returns.
\subsection{Portfolio Variance with Homoskedastic Isotropic Returns}
For simplicity of exposition, consider the case of homoskedastic isotropic returns. With this choice
\begin{equation}\eqnlab{isovar}
V_P= \sigma^2\frac{1+(N-1)\rho}{N}\ \Rightarrow\ N^*=\frac{N}{1+(N-1)\rho}.
\end{equation}
Note that positive correlation reduces the effective degrees of freedom and negative correlation increases it,

However, \eqnref{isovar} has two interesting limits. Firstly, the requirement that $V_P\ge0$ requires that
\begin{equation}
    1+(N-1)\rho\ge0\ \Rightarrow\ \rho\ge-\frac{1}{N-1}.
\end{equation}
This means that perfect anti-correlation, or the case when all asset returns are completely anti-correlated with all asset returns with $\rho=-1$, is only possible for a universe of two assets. For three assets $\rho\ge-1/2$, for four $\rho\ge-1/3$ etc. In the ``large portfolio'' limit 
\begin{equation}
\lim_{N\rightarrow\infty}(\min\rho)=0,
\end{equation}
meaning that negative correlation cannot exist in arbitrarily large portfolios of assets with isotropic returns.

Secondly, for very large portfolios, the effective degrees of freedom is limited to the reciprocal of the correlation coefficient:
\begin{equation}\eqnlab{doflim}
    \lim_{N\rightarrow\infty}N^*=\frac{1}{\rho}.
\end{equation}
Thus, in such a ``large'' portfolio \eqnref{rhovals} applies:
\begin{equation}\eqnlab{rhovals}
    0\le N^*\le\frac{1}{\rho}.
\end{equation}
\subsection{The Asymptotic Normality of Large Portfolio Returns}
A result generally held to be true for equity portfolios is that ``arbitrarily large'' portfolios have Normally distributed returns. This is viewed the consequence of the \textit{Central Limit Theorem}\cite[pp. 310--312]{kendall1999vol1} directly applied to a sum of random variables that represents the return of a portfolio decomposed into individual asset returns. However, empirical work shows that many equity indices, despite the fact that they are literally composed as a weighted sum of asset values, do not exhibit Normally distributed returns.\footnote{See the author's prior works for an extensive discussion\cite{giller2022adventures,giller2022normal,giller2023countries}.}

If asset returns are well described by an isotropic correlation matrix then \eqnref{doflim} actually has dire consequences for this assumed Normal convergence of portfolio returns. Typical values of pairwise correlations are of order 20\% (see \secref{empirical}) which implies that the upper limit on $N^*$ would be around five! It is common Statistical practice to assume that the independent sample size sufficient for the error made by assuming the sampling distribution of statistic to be ``well approximated'' by the Normal distribution is around thirty,\footnote{This value is of widespread use within the educational and practical Statistics communities.} or that convergence in distribution is, in fact, quite rapid. If the empirical distribution of asset returns is isotropic in the way described here, then convergence to Normality in distribution for any reasonably popular stock market index cannot be achieved for these portfolios cannot contain more that around five effective degrees of freedom no matter their actual sizes, which are typically in the hundreds (\textsc{NASDAQ-100}, \textsc{S\&P~500}, \textsc{FTSE}, \textsc{Nikkei 225} etc.) or thousands (\textsc{S\&P Composite 1500}, \textsc{Russel 3000}, \textsc{Wilshire 5000} etc.). Even a portfolio of every listed equity on the planet could not exhibit Normally distributed portfolio returns under an isotropic covariance matrix if correlations values match those typically measured.

Taking the value $N^*\approx30$ as the target, \eqnref{doflim} implies that the maximum pairwise correlation must be around 3\%, or smaller, for convergence in distribution to occur. This value is plainly at variance with the stylized facts around the correlation of asset price returns.
\subsection{Eigendecomposition of Isotropic Correlation Matrices}\seclab{eigen}
A common starting point for \textit{Factor Analysis}, or the representation of multivariate random vectors in terms of the linear superposition of i.i.d.\ risk ``factors,'' is often \textit{Principal Components Analysis}\cite[ch. 8]{mardia1979}. This is a well known process in which a covariance matrix, $\Sigma$ is decomposed into $Q^TDQ$ where $Q$ is an orthogonal matrix constructed from the eigenvectors that solve $\Sigma\boldsymbol{x}=\lambda\boldsymbol{x}$, for scalar $\lambda$, and $D$ is a diagonal matrix formed from the vector of the $N$ eigenvalues, i.e. $D=\mathop{\mathrm{diag}}\boldsymbol{\lambda}$\cite[ch. 4]{arfken1999mathematical}.\footnote{Some authors use $QDQ^T$.} This procedure, which replaces the $N$ original random variables with $N$ linear combinations of them that are statistically independent, i.e. the vector $Q\boldsymbol{x}$ which has diagonal covariance matrix $D$, can \textit{always} be executed because the matrix $\Sigma$ is symmetric positive definite by definition, and so does not generate any new information, it merely partitions the variance of $\boldsymbol{x}$ into a useful structure by executing a coordinate rotation in $\mathbb{R}^N$.

It is interesting to seek such an decomposition of the the homoskedastic isotropic covariance matrix, $\sigma^2G_N$, as this will give insight into the forms of factor models that the structure supports. This is equivalent to the decomposition of $G_N$ itself, as the eigenvalue problem is scale invariant.\footnote{Specifically, if $\boldsymbol{x}$ is a solution to $A\boldsymbol{x}=\lambda\boldsymbol{x}$ then $\boldsymbol{x}$ is also a solution to $B\boldsymbol{x}=\lambda\boldsymbol{x}$, for $B=kA$ with non-zero scalar $k$.} It can be shown by induction\footnote{See appendix \ref{section:aeigen}.} that matrix $G_N$ has $N-1$ eigenvalues equal to $1-\rho$ and one eigenvalue equal to $1+(N-1)\rho$. Because this system has $N-1$ \textit{degenerate} eigenvalues, the associated eigenmatrix is not necessarily orthogonal but it may be made orthogonal by the Gram-Schmidt procedure\cite[pp. 516--223]{arfken1999mathematical}.

For $N=1$ the eigenmatrix is trivially the identity matrix of dimension $1$ and the associated eigenvalue is also 1. For $N>1$, an orthogonal matrix, $Q_N$, of the form
\begin{equation}
    \begin{pmatrix}
    1/\sqrt{N}&0&0&\cdots&0\\
    0&1/\sqrt{2}&0&\cdots&0\\
    0&0&1/\sqrt{6}&\cdots&0\\
    \vdots&\vdots&\vdots&\ddots&\vdots\\
    0&0&0&\cdots&1/\sqrt{N(N-1)}
    \end{pmatrix}
    \begin{pmatrix}\eqnlab{Qn}
        \phantom{-}1&\phantom{-}1&\phantom{-}1&\cdots&1\\
        -1&\phantom{-}1&\phantom{-}0&\cdots&0\\
        -1&-1&\phantom{-}2&\cdots&0\\
        \phantom{-}\vdots&\phantom{-}\vdots&\phantom{-}\vdots&\ddots&\vdots\\
        -1&-1&-1&\cdots&N-1
    \end{pmatrix}
\end{equation}
may be constructed.\footnote{I have factored out the normalization to expose the structure of the underlying matrix.}
\subsection{The Factor Model Associated with Homoskedastic Isotropic Covariance}
\seclab{isofac}
From a finance perspective, the first row of $Q_N$, when multiplied into a returns vector, will generate the returns proportional to those of an equal-weighted portfolio (a ``market'' factor) and the other rows represents the returns of every possible ``spread-trade'' that features a long position in the final asset included in the portfolio and short positions in all of the others included. Note that this is not \textit{all possible} spread trades, but a specific subset of them which always involves short positions in assets $[1,M-1]$ and a long position in asset $M$, for asset index $M\in[2,N]$. Since the index of this long asset is entirely arbitrary,\footnote{The structure of the problem would not change were we to arbitrarily re-label all of the assets with any permutation of the possible labels.} it is unreasonable to infer that these spread-trade portfolios represent the returns of meaningful risk factors, which is unlike that of the market factor which treats all assets equally. Heuristically, this can then be used to partition the risk represented by the covariance matrix $G_N$ into that arising from one common risk factor and that the rest of the risk is ``idiosyncratic.'' i.e. To make the choice that the number of meaningful Principal Components within the cross-section of returns is just one.\footnote{Which is exactly the choice that would be made by the traditional algorithm of the examination of an ``elbow'' plot of the eigenvalue spectrum, as first suggested by Thorndike\cite{thorndike1953family}.}

Consider an equal-weighted portfolio with weights $1/N$ composed from assets with a homoskedastic isotropic covariance matrix. The total systematic risk associated with that portfolio is clearly given by the variance due to the single common factor, which is
\begin{equation}
    V_S=\sigma^2\frac{1+(N-1)\rho}{N^2}.
\end{equation}
The total residual idiosyncratic risk is given by the sum of the remaining eigenvalues, or $(N-1)(1-\rho)$, with the same constant of proportionality. i.e.
\begin{equation}
    V_R=\sigma^2\frac{(N-1)(1-\rho)}{N^2}.
\end{equation}
\pdffigure{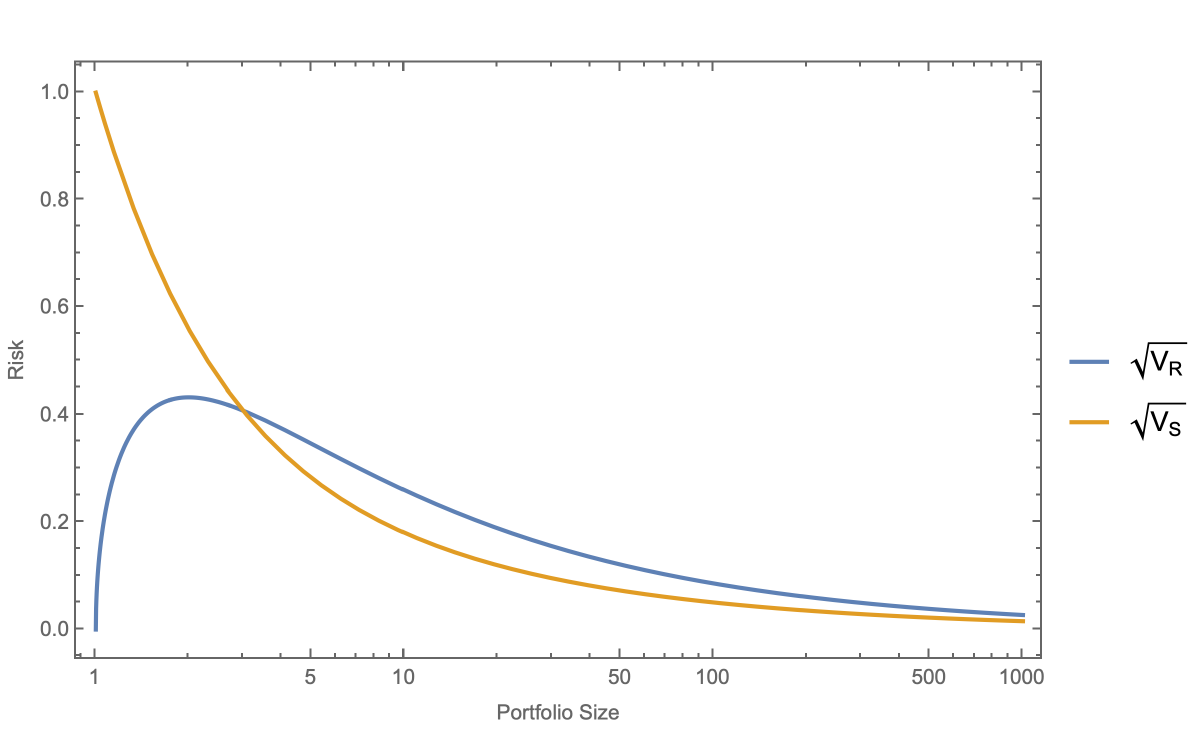}{The limiting behaviour of the total residual risk and total systematic risk of an equal-weighted portfolio when assets are described by an isotropic covariance matrix.}{.6,bb=0 0 7.96in 4.47in}{0}

The limiting behaviour of these measures is shown in \figref{risk}.
In the large portfolio limit for this system the residual risk \textit{does not vanish} except in the case of perfect correlation, $\rho=1$, in which case there is only common factor risk. 
\begin{equation}\eqnlab{resr2}
    \lim_{N\rightarrow\infty}\frac{V_R}{V_S}=\frac{1-\rho}{\rho}.
\end{equation}
This is a number of order unity for typical values of $\rho$.\footnote{For example, it is $3$ when $\rho=25\%$, $1$ when  $\rho=50\%$, and $1/3$ when $\rho=75\%$.} Thus, this single factor model is not equal to the Capital Asset Pricing Model of Sharpe \textit{et al.} because it is not possible to eliminate residual risk through diversification and, therefore, there must remain a premium to be paid to investors who take residual risk within their portfolios.
\subsection{The Effective Degrees of Freedom in Standard Factor Models}\seclab{factor}
Principal Components Analysis, as discussed \vpageref[above]{section:eigen}, becomes Factor Analysis when a scheme is introduced to exclude some of the $N$ components, leaving a vector, $\boldsymbol{f}_t$, in $\mathbb{R}^K$ of independent risk factors and the remaining variance explained by a random variable $\boldsymbol{\varepsilon}_t$ in $\mathbb{R}^N$ for which all members are independent of each other \textit{and} of $\boldsymbol{f}_t$ as well.

This can be expressed via the familiar linear additive noise model
\begin{align}
    \boldsymbol{r}_t&=\boldsymbol{\mu}+B\boldsymbol{f}_t+\boldsymbol{\varepsilon}_t\\
    \mathrm{where}\ \mathbb{E}[\boldsymbol{f}_t]&=\boldsymbol{0}_K\ \mathrm{and}\ 
    \mathbb{E}[\boldsymbol{\varepsilon}_t]=\boldsymbol{0}_N\ \Rightarrow\ \mathbb{E}[\boldsymbol{r}_t]=\boldsymbol{\mu}\\
    \mathrm{and}\ \mathbb{V}[\boldsymbol{r}_t]&=B\mathbb{V}[\boldsymbol{f}]B^T+\mathbb{V}[\boldsymbol{\varepsilon}]\ \mathrm{with}\ \mathbb{V}[\boldsymbol{f},\boldsymbol{\varepsilon}]=0.
\end{align}
It also is common in this decomposition to set the factor variances themselves to unity, which merely transfers the scale factor due to the particular values along the diagonal of $\mathbb{V}[\boldsymbol{f}_t]$ into the $N\times K$ factor loading matrix $B$ and, otherwise, is a trivial change to the model. This means that $\mathbb{V}[\boldsymbol{f}]=I_K$, where $I_K$ is the identity matrix of dimension $K$. Let $\mathbb{V}[\boldsymbol{\varepsilon}_t]=\mathcal{S}^2$, where $\mathcal{S}^2$ is a diagonal matrix of dimension $N$ with the idiosyncratic variance of each stock along the diagonal. Thus
\begin{equation}
    \mathbb{V}[\boldsymbol{r}_t]=BB^T+\mathcal{S}^2.
\end{equation}

As \vpageref[above]{section:isofac}, consider an equal-weighted portfolio. The portfolio variance is
\begin{equation}\eqnlab{fpvar}
    V_P=\frac{\boldsymbol{1}_N^TBB^T\boldsymbol{1}_N+\mathop{\mathrm{tr}}\mathcal{S}^2}{N^2}=\frac{\mathop{\mathrm{gs}}BB^T+\mathop{\mathrm{tr}}\mathcal{S}^2}{N^2}.
\end{equation}

The factor loadings matrix, $B$, may be written in terms of a set of $N$ factor loadings vectors, $\{\boldsymbol{b}_i\}$, of dimension $K$. i.e.
\begin{equation}
    B=\begin{pmatrix}
        \boldsymbol{b}_1^T\\
        \vdots\\
        \boldsymbol{b}_N^T
    \end{pmatrix}
    \;\Leftrightarrow\;
    B^T=\begin{pmatrix}
        \boldsymbol{b}_1&\dots&\boldsymbol{b}_N
    \end{pmatrix}.
\end{equation}
Thus the outer product $BB^T$ can be seen to be the $N\!\times\!N$ matrix of all inner products of the $K$ dimensional loadings vectors:
\begin{align}
    BB^T&=\begin{pmatrix}
        \boldsymbol{b}_1^T\boldsymbol{b}_1&
            \boldsymbol{b}_1^T\boldsymbol{b}_2&
            \dots&
            \boldsymbol{b}_1^T\boldsymbol{b}_N\\
        \boldsymbol{b}_2^T\boldsymbol{b}_1&
            \boldsymbol{b}_2^T\boldsymbol{b}_2&
            \dots&
            \boldsymbol{b}_2^T\boldsymbol{b}_N\\
        \vdots&\vdots&\ddots&\vdots\\
        \boldsymbol{b}_N^T\boldsymbol{b}_1&
            \boldsymbol{b}_N^T\boldsymbol{b}_2&
            \dots&
            \boldsymbol{b}_N^T\boldsymbol{b}_N
    \end{pmatrix}\eqnlab{BBT}\\
\Rightarrow\;\mathop{\mathrm{gs}}BB^T&=
(\boldsymbol{b}_1^T+\dots+\boldsymbol{b}_N^T)\times
(\boldsymbol{b}_1+\dots+\boldsymbol{b}_N).\eqnlab{gsBBT}
\end{align}
Let $\overline{\boldsymbol{b}}$ represent the arithmetic mean of the loadings vectors. \Eqnref{gsBBT} may then be written
\begin{equation}
    \mathop{\mathrm{gs}}BB^T=N^2\,\overline{\boldsymbol{b}}^T\overline{\boldsymbol{b}}.
\end{equation}
Furthermore, let $\{s_i\}$ be the diagonal elements of $\mathcal{S}$ and write $\overline{s^2}$ for the mean idiosyncratic variance. \Eqnref{fpvar} may then be written
\begin{equation}
V_P=\overline{\boldsymbol{b}}^T\overline{\boldsymbol{b}}+\frac{\overline{s^2}}{N},
\end{equation}
where we identify the two terms as the total systematic variance, $V_S$, and the total residual variance, $V_R$. $V_S$ is a scalar independent of portfolio size so, in the large portfolio limit,
\begin{equation}
\lim_{N\rightarrow\infty}V_P=V_S\;\Rightarrow\;\lim_{N\rightarrow\infty}\frac{V_R}{V_S}=0.
\end{equation}
The contribution from idiosyncratic variance to portfolio variance is diversified away. This is, of course, a foundational result in both the \textit{Capital Asset Pricing Model}\cite{sharpe1964capm} and \textit{Arbitrage Pricing Theory}\cite{ross2013arbitrage}. It is fundamentally different from the result of \eqnref{resr2} and indicates that it should be possible to distinguish empirically between these two theories on the cross-section of equity returns.

For this model, the total independent variance in an equal-weighted portfolio of assets with with $K$ factors, $V_I$ as defined in \eqnref{portvar}, is the trace of the covariance matrix divided by the square of the portfolio size. i.e.
\begin{equation}
    V_I=\frac{\mathop{\mathrm{tr}}(BB^T+\mathcal{S}^2)}{N^2}=\frac{\mathop{\mathrm{tr}}BB^T}{N^2}+\frac{\overline{s^2}}{N}.
\end{equation}
From the properties of the trace,\footnote{This may be seen directly by inspection of \eqnref{BBT}.} $\mathop{\mathrm{tr}}BB^T$ is equal to the sum of the squares of all of the elements of matrix $B$, or $\sum_{ij}b_{ij}^2$ for matrix elements $\{b_{ij}\}.$ Therefore
\begin{equation}
    V_I=\overline{b^2}+\frac{\overline{s^2}}{N},
\end{equation}
with $\overline{b^2}=(\mathop{\mathrm{tr}}BB^T)/N^2$.  The effective degrees of freedom (\eqnref{Nstar}) is then
\begin{equation}\eqnlab{Nstarfactor}
    N^*=\frac{\overline{b^2}N+\overline{s^2}}{\overline{\boldsymbol{b}}^T\overline{\boldsymbol{b}}+\overline{s^2}/N}
    =N\frac{\overline{b^2}N+\overline{s^2}}{\overline{\boldsymbol{b}}^T\overline{\boldsymbol{b}}N+\overline{s^2}}.
\end{equation}
\subsubsection{Important Limits of $N^*(N)$ for $K$ Linear Factors}\seclab{faclimit}
\Eqnref{Nstarfactor} has four important limiting cases.
Firstly, for arbitrarily large portfolios with less factors than assets (i.e. $K<N$) the effective number of degrees of freedom increases without bound, however it tends towards a value independent of the residual risk:
\begin{equation}
    N^*\xrightarrow[N\to\infty]{}\frac{\overline{b^2}}{\overline{\boldsymbol{b}}^T\overline{\boldsymbol{b}}}N.
\end{equation}
Secondly, for completely independent assets (i.e. $K=0$) it is merely equal to the total number of assets in the portfolio:
\begin{equation}
\lim_{B\rightarrow0}N^*=N,
\end{equation}
which is the same result as that in \eqnref{isovar} for $\rho\rightarrow0$.
Thirdly, for the case of a full Principal Components Analysis where $K=N$, we have $B=D^{1/2}Q$ since $\Sigma=QDQ^T$ so $BB^T=D$ meaning that $\mathop{\mathrm{gs}}BB^T=\mathop{\mathrm{tr}}BB^T$ and $\overline{s^2}=0$ in the above expressions. Therefore
\begin{equation}
    \lim_{K\rightarrow N}N^*=N.
\end{equation}
Finally consider the case when there are $K$ factors but none of the factors are dominant so that the loadings of all stocks onto any factor are similar. That is $b_{ij}\approx b$, for some constant $b$. This also covers a C.A.P.M.\ type structure ($K=1$) in which all stocks have similar ``betas.'' For this case
\begin{equation}\eqnlab{largeNKfactors}
    N^*\approx\frac{b^2N+\overline{s^2}}{b^2K+\overline{s^2}/N}\;
    \Rightarrow\;N^*\xrightarrow[N\to\infty]{}\frac{N}{K}.
\end{equation}
\subsubsection{The Relationship between Factor Loading Matrix Means}\seclab{facmean}
\Eqnref{Nstarfactor} may be thought of as an expression for $N^*$ in terms of four parameters: $N$, $\overline{s^2}$, $\overline{b^2}$ and $\overline{\boldsymbol{b}}^T\overline{\boldsymbol{b}}$. Since the latter two are \textit{both} themselves functions of the entire factor loading matrix, $B$, clearly their values are not independent: one is a \textit{mean-squared} and the other a \textit{squared-mean} of the matrix elements $b_{ij}$. One should expect a ``triangular'' relationship between them, which is as follows.\footnote{This derivation follows the form of well known proofs that the population variance is non-negative.}

Consider the real valued matrix $B$ of dimension $N\times K$, with elements $b_{ij}$, and, as above, let the column means be
\begin{equation}
    \overline{b}_j=\frac{1}{N}\sum_{i=1}^Nb_{ij}.
\end{equation}
From the definition of $B$ it follows that
\begin{equation}
    \left(b_{ij}-\overline{b}_j\right)^2\ge0.
\end{equation}
Therefore the mean over all matrix elements of this quantity must also be non-negative.
\begin{equation}
    \frac{1}{KN}\sum_{j=1}^K\sum_{i=1}^N\left(b_{ij}-\overline{b}_j\right)^2\ge0.
\end{equation}
Expanding the square and summing over the assets ($i$) and factors ($j$) then gives:
\begin{equation}
    \overline{b^2}\ge\frac{1}{K}\overline{\boldsymbol{b}}^T\overline{\boldsymbol{b}}.
\end{equation}
That is the mean of the squared matrix elements is not smaller than the mean of the squared column means. 
\subsection{Summary of Results}
From the above we see that an homoskedastic isotropic covariance matrix may be though of as generating a single factor model, in that there is clearly a ``market'' factor that may be composed from asset returns and there is a factor-replicating portfolio (the equal-weighted portfolio) associated with it, but this is not a ``true factor model,'' as is usually defined, since the residual returns are not independent and residual risk cannot be fully diversified away.

The functional form of $N^*(N)$ presents a signature of the nature of the covariance of equity returns that is independent of distributional choice and has a specific form that differs between isotropic covariance and linear factor models, thus it may be used to discriminate between these two hypothesis in real data. For a isotropic covariance models $N^*(N)$ is asymptotically a constant given by $1/\rho$ whereas for linear factor models $N^*(N)$ is asymptotically proportional to $N$ and divergent.
\section{Empirical Measurements of the Degrees of Freedom within Major Market Indices}\seclab{empirical}
Results are presented for the adjusted daily returns of members of the S\&P~500 Index since the last index re-balance. Restricting the analysis to this period removes survivorship bias but it does, potentially, introduce temporal bias, since the data is restricted to a short, recent, history of asset price returns. The data used is available publicly and full details of the code and data sources are given in Appendix~\ref{section:data}. At the time of writing (October, 2024) the last index re-balance was on September 30th., 2024, and there were 503 stocks in the index.
\subsection{Exploratory Data Analysis}
To gain insights into the likely values of pairwise correlations between assets
a simple random sampling experiment was executed. Pairs of assets were selected at random from the available universe of index members and the correlation of their daily adjusted returns computed. For an index of size $N_\mathrm{max}$ there are $N_\mathrm{max}!/\{2!(N_\mathrm{max}-2)!\}$ ways of picking such index pairs. This gives a total of 126,253 possible index pairs to examine, however just 5,000 random trials were selected. The data measured are exhibited in the histogram in \figref{rho}. This shows a mean Pearson correlation coefficient of $+17\%$ and broad dispersion of data about that value, with the range of values spanning from $-80\%$ to close to $+100\%$. The sample is visibly left-skewed.

As the homoskedastic isotropic correlation model hypothesizes that these correlation coefficients are are the same value, and that the dispersion seen arises \textit{purely} from sampling variation, it is useful to understand what the scale of that dispersion would be under the hypothesis of identical correlation between all pairs. In his work on the sampling distribution of the correlation coefficient, Fisher\cite{fisher1915frequency} introduced simple ``Normalizing'' transformation given by the inverse hyperbolic tangent, $\mathop{\mathrm{atanh}}\rho$. It was also shown that this transformed coefficient was rapidly asymptotically Normal with variance $1/(N_\mathrm{obs}-3)$, for sample size $N_\mathrm{obs}$.
\pdffigure{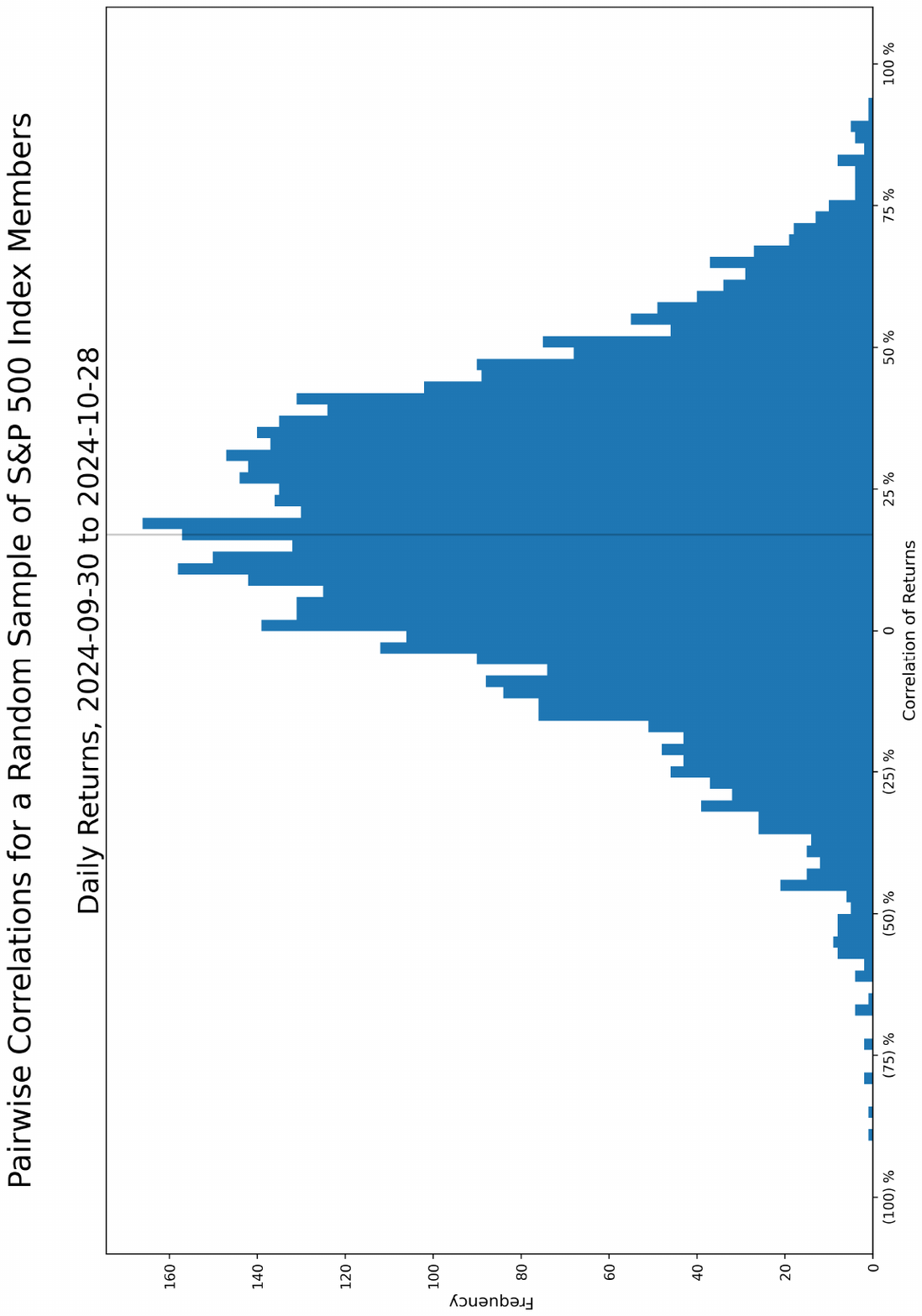}{The distribution of the correlation of adjusted daily returns for a random sample of index member pairs.}{0.45,bb=0 0 8.5in 11in}{270}
\pdffigure{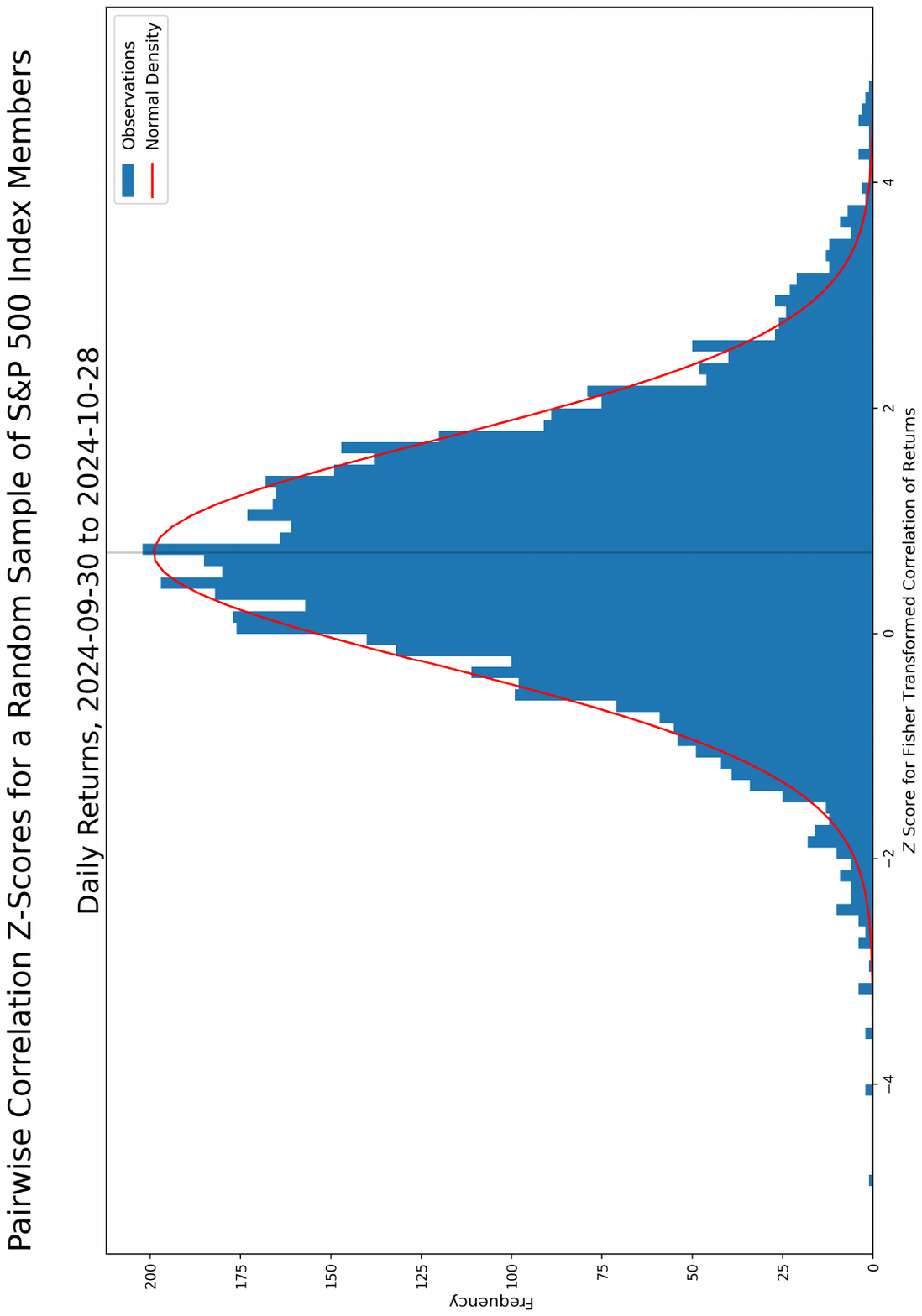}{The distribution of the $Z$ score for the Fisher transformed  correlation of adjusted daily returns for a random sample of index member pairs.}{0.45,bb=0 0 8.5in 11in}{270}

A better context on whether the data shown in \figref{rho} is consistent with being generated from a common correlation coefficient may therefore be gained from \figref{fisher}, where the same data is plotted but transformed through $Z=\sqrt{N_\mathrm{obs}-3}\mathop{\mathrm{atanh}}\rho$. Also shown is a Normal distribution curve with a mean given by $\mathop{\mathrm{atanh}}\overline{\rho}\approx0.22$ and a standard deviation of one. To the eye this data does seem very consistent with the unit Normal proposed, although it does fail a Kolmogorov-Smirnov test\cite[pp. 269--271]{kstest} with 5\% confidence having a two-sided test statistic of $D=0.021$ and a $p$ value of $0.018$. However, given the nature of financial data, it doesn't seem that much weight should be placed on such a weak result.
\subsection{Experimental Design}\seclab{method}
To evaluate the empirically observed relationship between effective degrees of freedom and portfolio size is quite straightforward. 

To evaluate the functional relationship $N^*(N)$ for the selected index it is necessary to explore a range of values for $N$ and compute $N^*$ as given in \eqnref{Nstar} for each $N$ considered. To prevent biases due to analysts choices a random sampling scheme was executed as follows:

\begin{enumerate}[label=\arabic*:]
    \item select $N$ at random (and uniformly) from $[1,N_\mathrm{max}]$, where $N_\mathrm{rmax}$ is the number of securities in the index;
    \item select $N$ securities at random from the list of index members given;
    \item for this set of securities compute the individual variances of daily returns, $\{\sigma^2_i\}$, and for an equal-weighted portfolio;
    \item from these data $V_I$ and $V_P$ may be evaluated, and ultimately $N^*(N)$.    
\end{enumerate}
This whole procedure was then repeated a 1,000 times. As there are $2^{N_\mathrm{max}}-1$ ways of picking between one and $N_\mathrm{max}$ securities from an index of this size, so it is computational infeasible to explore every portfolio in the population. and so an experiment based on random sub-sampling must be pursued.\footnote{At the time of writing (October, 2024) there are $503$ assets in the S\&P~500 index and $2^{503}-1\approx2.6\times10^{151}$.} 
\subsection{Results of Experiment}
The results of this experiment are shown in \figref{nstarn}. Although there is substantial sampling error for medium sized portfolios, the data appears to be consistent with the functional form expected for an homoskedastic isotropic covariance matrix. For this data set, $\hat{\rho}$ under this assumption, may be estimated as
\begin{equation}
    \hat{\rho}=\frac{N_\mathrm{max}-N^*}{(N_\mathrm{max}-1)N^*}.
\end{equation}
\pdffigure{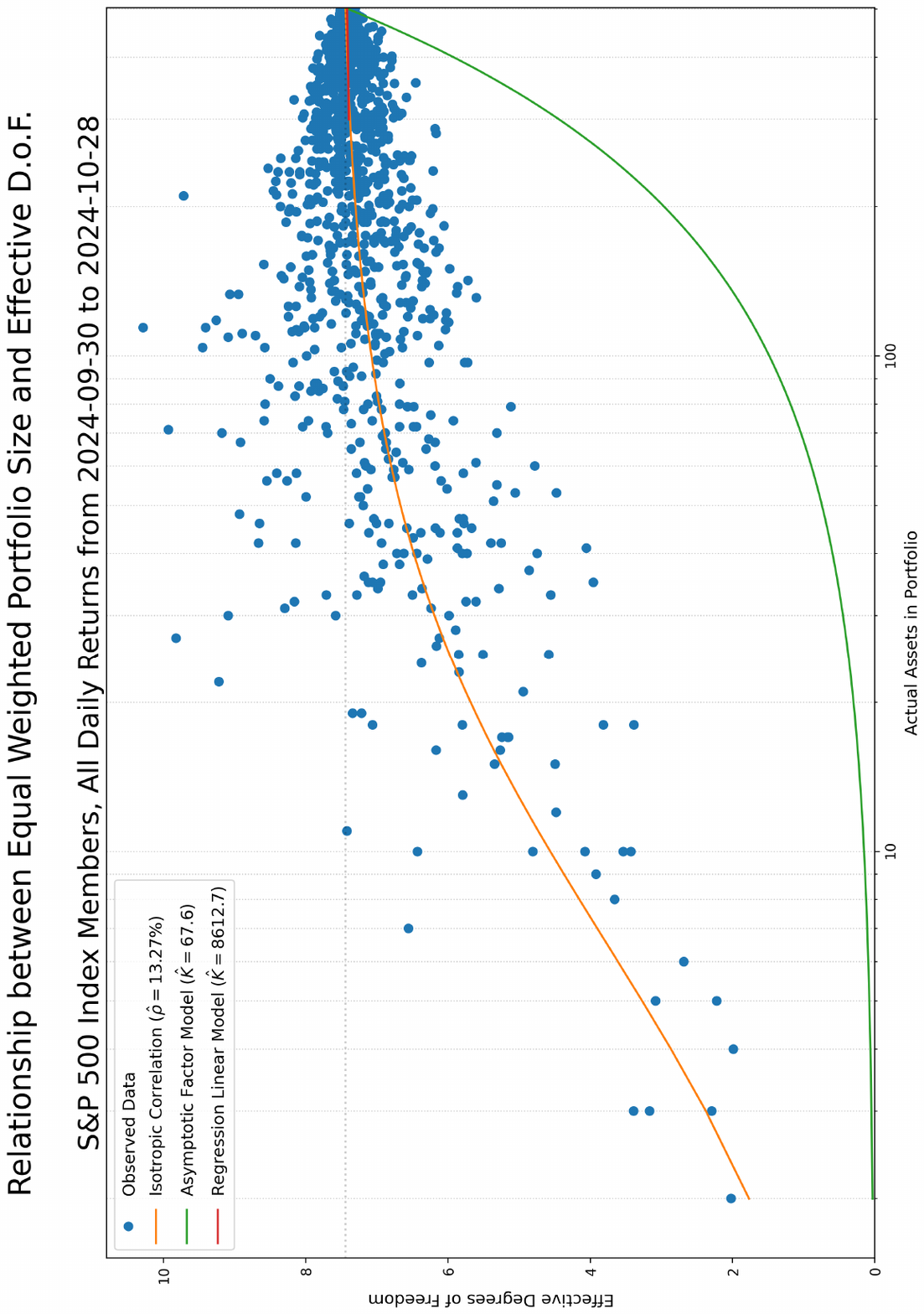}{A scatter plot of $N^*$ vs.\ $N$ for 1,000 portfolios formed according to the method described in \secref{method}. The orange line represents the curve expected for an homoskedastic isotropic covariance matrix, with the correlation chosen to be consistent with the value $N^*(N_\mathrm{max})$, and the green line represents the curve expected for a linear factor model with the number of factors chosen to be consistent with the same measure on the assumption of a homogeneous factor loading matrix.}{0.45,bb=0 0 8.5in 11in}{270}

For the S\&P~500 experiment, the maximum effective d.o.f., for an equal-weighted portfolio of size $503$ is measured to be $7.44$. This gives $\hat{\rho}=(503-7.12)/(502\times7.12)\approx13.27\%$. The functional form of $N^*(N)$ then follows immediately from this single estimate and is drawn as an orange line in the figure. The curve is not ``fitted'' across the span of the data, it is \textit{extrapolated} from the measured terminal value. The agreement with the other $O(1000)$ points in the sample created by the experiment is quite clear.

However, this is not sufficient to \textit{rule out} the applicability of other models, for they may generate the same kind of curve. In that situation, the best that could be said was that the data is not inconsistent with a heteroskedastic isotropic correlation model. To that end two further curves are drawn on \figref{nstarn}.

The first of these is the curve expected in the ``large portfolio, $K$ linear factors'' limit of \eqnref{largeNKfactors}, or $N^*\simeq N/\hat{K}$. With the data computed this gives $\hat{K}\approx502/7.44=68$ and the curve is drawn with a green line. It is clear that, apart from the final point that it is required to fit, the data in no way resembles this asymptotic form.

A second approach is to approximate \eqnref{Nstarfactor} with a form that assumes that $N$ is sufficiently large that the term in $\overline{s^2}/N$ in the denominator may be neglected, leaving a simple linear relationship
\begin{equation}
    N^*(N)\simeq\frac{\overline{b^2}}{\overline{\boldsymbol{b}}^T\overline{\boldsymbol{b}}}N+\overline{s^2}.
\end{equation}
The two scalars in this expression may be estimated by ordinary least squares in the large N region where the observed relationship between $N^*$ and $N$ does not seem to exhibit much curvature.

A fit over the region $N\in[300,502]$ produces the red line shown in \figref{nstarn}. This line is not visually separable from the orange curve due to the heteroskedastic isotropic model. The fit has an $R^2$ of just $0.1\%$ and an $F$ statistic of $0.28$ for $1$ and $409$ degrees of freedom. This is not a significant regression and the slope coefficient of $0.0001\pm0.0002$ has a $t$ score of $0.528$ meaning that there is no evidence that the null hypothesis of zero slope should be rejected. Nevertheless, this estimated slope implies $\hat{K}\approx8,613$, which is a factor of $17$ more than the maximum number of assets for which a portfolio may be composed. It does not seem reasonable to conclude that this data is suggesting the returns of a portfolio formed from the members of the S\&P~500 index has this many degrees of freedom. Under the common interpretation of the regression results, which is that the slope coefficient is not significantly distant from zero, this requires the conclusion that the data is not following a linear factor model. Interpreting the fitted constant, $7.358\pm0.088$ as an estimator of the common correlation coefficient, $\rho$, according to the limit of \eqnref{doflim}, this gives an estimate $\hat{\rho}=13.6\%$, similar to the other measurements. The $t$ score of the untransformed intercept is $83.2$ and it is most definitely not consistent with zero!
\section{Portfolio Selection with Isotropic Covariance Matrices}
\subsection{Mean-Variance Optimization}
After Markowitz\cite{harry1952markowitz}, the mean-variance optimal portfolio, $\hat{\boldsymbol{h}}$, with $N$ securities  may be written:
\begin{equation}
   \hat{\boldsymbol{h}}=\argmin_{\boldsymbol{h}\in\Omega}
   \left(\boldsymbol{h}^T\boldsymbol{\alpha}-\lambda\boldsymbol{h}^T\Sigma\boldsymbol{h}\right),
\end{equation}
where $\Omega\subseteq\mathbb{R}^N$ represents a space of feasible portfolios, $\boldsymbol{\alpha}$ is a vector of expected returns and $\Sigma$ is the covariance matrix of the assets. In this equation the Lagrange multiplier\cite[pp. 945--950]{arfken1999mathematical}, $\lambda$, can also be thought of as the ``market price'' of risk. It is straightforward to show by calculus that
\begin{equation}
    \hat{\boldsymbol{h}}=\frac{\Sigma^{-1}\boldsymbol{\alpha}}{2\lambda}.
\end{equation}
The specific value of $\lambda$ may be adjusted to ensure the optimal portfolio, $\hat{\boldsymbol{h}}$, lies within the feasible region $\Omega$, or may be set by other criteria if it is not bound by the constraints.\footnote{For example, Thope\cite{thorp2011kelly} shows that $\lambda=1/2$ is asymptotically equivalent to the Kelly Criterion\cite{kelly1956new} under the assumption of Normally distributed returns.}

Assuming the form of \eqnref{Vdef} for $\Sigma$ gives
\begin{equation}
\Sigma^{-1}=S_N^{-1}G_N^{-1}S_N^{-1}\;\Rightarrow\;\hat{\boldsymbol{h}}=\frac{S_N^{-1}G_N^{-1}S_N^{-1}\boldsymbol{\alpha}}{2\lambda}.
\end{equation}
From the decomposition of $H_N$ given in the proof of 
\theref{inv2}, and the associated factorization of $f_N$ as $(\rho-1)\{1+(n-1)\rho\}$, the inverse of $G_N$ may be written
\begin{equation}
    G_N^{-1}=\frac{I_N}{1-\rho}-\frac{\rho\boldsymbol{1}_N\boldsymbol{1}_N^T}{(1-\rho)\{1+(N-1)\rho\}}.
\end{equation}
As $S_N$ is a diagonal matrix the product $S_N^{-1}\boldsymbol{\alpha}$ simply represents the ``$Z$-Scores'' of the expected returns, or $\boldsymbol{z}$. Therefore
\begin{equation}\eqnlab{mvoiso}
    \hat{\boldsymbol{h}}=\frac{S_N^{-1}}{2\lambda(1-\rho)}\left\{\boldsymbol{z}
    -\frac{N\rho\overline{z}\boldsymbol{1}_N}{1+(N-1)\rho}
    \right\}
\end{equation}
where $N\overline{z}=\boldsymbol{1}^T_N\boldsymbol{z}$. In the large portfolio limit
\begin{equation}
    \lim_{N\rightarrow\infty}\hat{\boldsymbol{h}}=\frac{S_N^{-1}(\boldsymbol{z}-\overline{z}\boldsymbol{1}_N)}{2\lambda(1-\rho)},
\end{equation}
and for the homoskedastic case
\begin{equation}
    \lim_{N\rightarrow\infty}\hat{\boldsymbol{h}}=\frac{\boldsymbol{\alpha}-\overline{\alpha}\boldsymbol{1}_N}{2\lambda(1-\rho)\sigma^2}.
\end{equation}
The effect of the correlation is twofold: when $0<\rho<1$ the portfolio is scaled up relative to the case of no correlation; and, more interestingly, when portfolios are ``large'' the optimal strategy is to invest proportional to ``relative alpha,'' $\boldsymbol{\alpha}-\overline{\alpha}\boldsymbol{1}_N$, whereas when they are small the centering term, $\overline{\alpha}\boldsymbol{1}_N$ is suppressed relative to the ``outright alpha'' by a factor $\rho N/\{1+(N-1)\rho\}$. This is shown for various values of $\rho$ in \figref{centering}.
\pdffigure{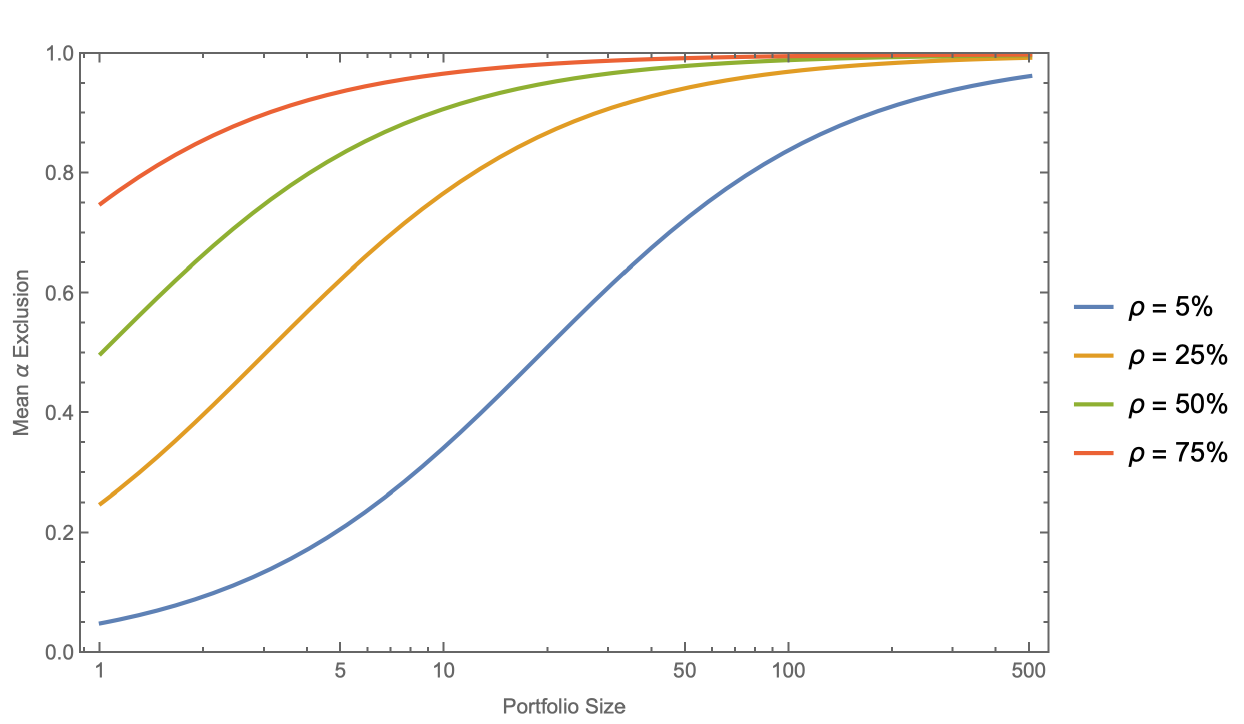}{The limiting behaviour of the factor scaling the ``mean-alpha'' term in a mean-variance optimal portfolio when assets are described by an isotropic covariance matrix.}{.6,bb=0 0 8.2in 4.47in}{0}
\subsection{Joint Distributions of Returns with Elliptic Symmetry}
In the author's earlier work asset allocation via negative exponential utility is investigated for distributions of returns with elliptical symmetry\cite{giller2004frictionless}. That is multivariate probability distributions distributions that can be written
\begin{equation}\eqnlab{gdef}
    f(g^2)\;\mathrm{where}\;g^2=\frac{(\boldsymbol{r}_t-\boldsymbol{\alpha})^T\Sigma^{-1}(\boldsymbol{r}_t-\boldsymbol{\alpha})}{\xi^2},
\end{equation}
$\Sigma$ is the covariance matrix of returns, and $\xi$ is a distribution dependent scalar.\footnote{The term $\xi$ is not introduced in the expression in the prior work, which is written in terms of a parameter matrix ``$\Sigma$'' that the covariance of returns is proportional to. In \textit{this} work $\Sigma$ represents the covariance matrix itself so $\xi$ is added here due to the notational collision.} For this system, it is shown that a negative exponential utility maximizer should chose a portfolio\footnote{In this expression, the term $\xi^2$ can be absorbed into the definition of $\lambda$ without affecting the solution.}
\begin{equation}
    \hat{\boldsymbol{h}}_t=\frac{\Sigma^{-1}\boldsymbol{\alpha}}{2\lambda\xi^2\Psi_{\!N/2}(\hat{x})}
\end{equation}
where $\Psi_\nu(x)$ is a special function that represents the ratio of two integral transforms of the radial distribution of returns and $\hat{x}$ is a value chosen to satisfy
\begin{equation}\eqnlab{elliptical}
    \hat{x}=x\;:\;x\Psi_{\!N/2}(x)=\sqrt{\boldsymbol{\alpha}^T\Sigma^{-1}\boldsymbol{\alpha}}.
\end{equation}

The portfolio represented by \eqnref{elliptical} can be thought of as equal to that delivered by canonical mean-variance optimization multiplied by a scaling function, $\Omega(Z)$, of the Mananalobis distance of the alpha defined by $Z^2=\boldsymbol{\alpha}^T\Sigma^{-1}\boldsymbol{\alpha}$ for distribution moments $\mathbb{E}[\boldsymbol{r}_t]=\boldsymbol{\alpha}$ and $\mathbb{V}[\boldsymbol{r}_t]=\Sigma$. $\Omega(Z)$ is in general a non-negative monotonic function of it's argument that takes unit value at the origin and is equal to that \textit{everywhere} for a multinormal distribution of returns. For leptokurtotic distributions of returns it is a \textit{decreasing} function of $Z$.

For isotropic correlation the value of $G_N^{-1}$ is known and so $Z$ may be explicitly written as
\begin{align}    Z^2&=\boldsymbol{\alpha}^TS_N^{-1}G_N^{-1}S_N^{-1}\boldsymbol{\alpha}\\
   &=\boldsymbol{z}^TG_N^{-1}\boldsymbol{z}\\
   &=\frac{N}{1-\rho}\left\{\overline{z^2}-\frac{\rho N}{1+(N-1)\rho}\overline{z}^2\right\},\eqnlab{Z2value}
\end{align}
where $\overline{z}^2$ and $\overline{z^2}$ are the squared-mean and mean-squared of the elements of $\boldsymbol{z}$ respectively.\footnote{The term suppressing the mean alpha in \eqnref{mvoiso} is also seen playing the same role in \eqnref{Z2value}.} In the large portfolio limit the term in in braces  converges to the variance of the elements of the specific sample value of $\boldsymbol{z}$ i.e.
\begin{equation}
    Z^2\xrightarrow[N\to\infty]{}\frac{N}{1-\rho}\left(\overline{z^2}-\overline{z}^2\right).
\end{equation}
\subsection{Solution with a Multivariate Laplace Distribution}
Specializing to the case of a multivariate Laplace distribution, as defined in prior work\cite{giller2024laplace}, the functional form of the $\Omega(Z)$ function is shown to be
\begin{equation}
    \Omega(Z)=\frac{\sqrt{1+4 Z^2/(N+1)}-1}{Z^2/(N+1)}
\end{equation}
and so the optimal portfolio for a negative exponential utility maximizer is
\begin{equation}
\hat{\boldsymbol{h}}=\frac{S_N^{-1}}{2\lambda(1-\rho)}\left\{\boldsymbol{z}
    -\frac{N\rho\overline{z}\boldsymbol{1}_N}{1+(N-1)\rho}
    \right\}\frac{\sqrt{1+4 Z^2/(N+1)}-1}{Z^2/(N+1)}.
\end{equation}
Care must be taken when examining the asymptotic behaviour of this expression as the $\Omega(Z)$ term should not be treated as having the na\"ive limit $\Omega\to2$ since $Z^2\propto N$ for large $N$.
\section{Conclusions}
Somewhat surprisingly\footnote{At least to the author, who has spent his professional career working within the A.P.T.\ framework.} this simple statistic, $N^*(N)$, which makes no distributional assumptions and merely relies on carefully studying how variance accumulates as assets are added to a portfolio, gives strong support to the idea that a heteroskedastic isotropic covariance model is a reasonable description of the equity cross-section, at least over recent history.\footnote{Clearly, a defect in this analysis is that it is restricted to only the most recent data. A further project to extend the empirical analysis presented herein to prior periods is ongoing.}

If the canonical models are \textit{insisted upon}, then the data seems to support a number of factors that is more aligned with A.P.T.\ based risk models where $K=O(100)$, as they are used on Wall Street, than the much smaller models, such as the Fama-French three factor model\cite{fama1973risk} or the Carhart model\cite{carhart1997persistence}, where $K=O(5)$. However, the reality is that the data does not seem to support the limiting behaviour, $N^*(N)\simeq N/K$, suggested by \textit{any} linear factor model for ``large'' $N=O(500)$, and these models do not appear to be supported by the data.

A key feature of the heteroskedastic isotropic correlation model, as developed in \secref{isofac}, is that it supports the concept of a ``market factor'' but it does not support the effectively complete elimination of residual risk through portfolio diversification. The ability to entirely remove residual risk is a foundational assumption for linear factor models such as C.A.P.M.\ or A.P.T., and so this structure is not equivalent to a reduced form ``nested within'' canonical risk models. It predicts something entirely different, which is that there is value to be obtained by picking stocks \textit{if} that activity is restricted to a ``large'' portfolio!

Finally, mean-variance optimal portfolios with isotropic correlation of asset returns should tilt away from outright alpha towards a standardized relative alpha as portfolio sizes get larger. This means that not only should the trader be picking stocks but that the extent to which they should be emphasizing relative value trades over tracking the ``market'' factor is a function of portfolio size. This is a somewhat heretical approach in modern times, and is driven by the failure of residual risk to be fully removed by diversification in large portfolios for these covariance structures. If one were to conjecture that the ability to ``manage'' a large portfolio is positively correlated, in some way, with assets under management then this result would suggest that small investors, i.e.\ ``retail investors,'' should concentrate on a whole-market index tracking fund whereas large investors should switch towards individual investment opportunities. Neither strategy is ``right,'' but the choice of index-tracking vs.\ stock picking is definitely linked to the ability to manage a large portfolio.\footnote{Where ``manage'' means accurately determine expected returns and expected variance.} This conclusion is not affected by non-Normal distributions of returns if they possess elliptic symmetry because, in this case, the modification to the optimal portfolio is to introduce a non-linear scale factor that is applied equally to all assets.
\appendix
\section{Solution of the Eigenvalue Problem and Inverse for Isotropic Correlation Matrices of any Size}\seclab{aeigen}
In \secref{eigen} it is asserted that the eigendecomposition of an isotropic correlation matrix has the particular form given there. In the following this is proved by induction.
\begin{definition}
    An isotropic correlation matrix, $G_N$, is a real valued symmetric 
    matrix of dimension $N\ge1$ with all diagonal entries equal to unity and all other entries equal to a constant $\rho\in[1/(1-N),+1]$ for $N>1$. $G_1$ is trivially the $1\!\times\!1$ matrix with element $1$.
\end{definition}
\begin{lemma}\lemlab{toprow}
    $G_N$ has an eigenvalue given by $1+(N-1)\rho$ and the associated eigenvector is proportional to the vector of ones of dimension $N$, or $\boldsymbol{1}_N$.
\end{lemma}
\begin{proof}
    For any square matrix, $A$, of dimension $N$, the product $A\boldsymbol{1}_N$ is a vector with each element equal to the sum of the corresponding row in $A$. This follows from the definitions of the matrix product and $\boldsymbol{1}_N$. The row-sums of $G_N$ are all $1+(N-1)\rho$ by definition. Therefore
    \begin{equation}
        G_N\boldsymbol{1}_N=\{1+(N-1)\rho\}\boldsymbol{1}_N.
    \end{equation}
\end{proof}
\noindent Specifically, consider the 1-dimensional case:
\begin{equation}
    G_1\boldsymbol{x}_1=\lambda_1\boldsymbol{x}_1.
\end{equation}
As $G_1=I_1$, where $I_N$ is the identify matrix of dimension $N$, this equation becomes $\boldsymbol{x}_1=\lambda_1\boldsymbol{x}_1$, which is trivially solved by $\lambda_1=1$ and $\boldsymbol{x}_1=\boldsymbol{1}_1$. These solutions are consistent with the theorem.
\begin{corollary}\corlab{orth}
    Let ${\boldsymbol{x}_i}$ for $i\in[1,N]$ represent the set of eigenvalues of $G_N$. \Lemref{toprow} shows that there exists an $\boldsymbol{x}_1\propto\boldsymbol{1}_N$. If we also require that the eigenvectors form an orthogonal basis in $\mathbb{R}^N$ then it immediately follows that $\boldsymbol{1}^T_N\boldsymbol{x}_i=0$ for all $i\in[2,N]$.
\end{corollary}
\noindent\Corref{orth} is useful, but it is not \textit{necessary} to make the assumption that the eigenvectors form an orthogonal basis and it's truth is conditional on that assumption. However \lemref{iterate} is generally true:
\begin{lemma}\lemlab{iterate}
    If $\boldsymbol{x}_N$ is an eigenvector of $G_N$ for $N>1$, and $\boldsymbol{x}_N$ is not proportional to $\boldsymbol{1}_N$, then there exists a vector $\boldsymbol{x}_{N+1}$, of dimension $N+1$, which is an eigenvector of $G_{N+1}$ and is identical to $\boldsymbol{x}_N$ apart from an added $0$ in the final row and has the same eigenvalue. Additionally, the sum of the elements of both of these vectors is zero.
\end{lemma}
\begin{proof}
If $G_N$ and $G_{N+1}$ are isotropic covariance matrices then it follows from their definition that $G_{N+1}$ may be written in the block-matrix form:
    \begin{equation}
        G_{N+1}=
        \begin{pmatrix}
        G_N&\rho\boldsymbol{1}_N\\
        \rho\boldsymbol{1}_N^T&1
        \end{pmatrix}.
    \end{equation}
The eigenvalue equation for $G_{N+1}$ may be written in block form as
\begin{equation}
    \begin{pmatrix}
        G_N&\rho\boldsymbol{1}_N\\
        \rho\boldsymbol{1}_N^T&1
    \end{pmatrix}
    \begin{pmatrix}
    \boldsymbol{x}_N\\x_{N+1}
    \end{pmatrix}=
    \lambda\begin{pmatrix}
    \boldsymbol{x}_N\\x_{N+1}
    \end{pmatrix}   
\end{equation}
for scalar $x_{N+1}$. Following the rules of matrix multiplication, this produces two equations
\begin{align}
    G_N\boldsymbol{x}_N+\rho x_{N+1}\boldsymbol{1}_N&
    =\lambda\boldsymbol{x}_N\eqnlab{GNa}\\
    \mathrm{and}\ \rho\boldsymbol{1}_N^T\boldsymbol{x}_N+x_{N+1}&
    =\lambda x_{N+1}.\eqnlab{GNb}
\end{align}
Let $\lambda_m$ for $m\in[2,N]$ be an eigenvalue of $G_N$ associated with an eigenvector which is not proportional to $\boldsymbol{1}_N$. Therefore $G_N\boldsymbol{x}_N=\lambda_m\boldsymbol{x}_N$. Substituting this into \eqnref{GNa} gives
\begin{equation}
    \rho x_{N+1}\boldsymbol{1}_N=\boldsymbol{0}_N\;
    \Rightarrow\;x_{N+1}=0\;\mathrm{or}\;\rho=0.
\end{equation}
For $\rho=0$, \eqnref{GNb} is trivially solved by $x_{N+1}=0$. For $\rho\ne0$, \eqnref{GNa} is solved by $x_{N+1}=0$. Therefore the eigenvalue equation for $G_{N+1}$ is solved by $\boldsymbol{x}_{N+1}$ with eigenvalue $\lambda_m$ when $x_{N+1}=0$. This proves the first part of the lemma. 

Now substitute $x_{N+1}=0$ into \eqnref{GNb}. This gives
\begin{equation}
    \rho\boldsymbol{1}_N^T\boldsymbol{x}_N=0,
\end{equation}
which proves the second part of the lemma when $\rho\ne0$. When $\rho=0$ the $G_N$ is equal to the identity matrix and the eigenvalue equation becomes $\boldsymbol{x}_N=\lambda\boldsymbol{x}_N$ which gives
\begin{equation}
\boldsymbol{1}_N^T\boldsymbol{x}_N=\lambda\boldsymbol{1}_N^T\boldsymbol{x}_N.
\end{equation}
This is solved by $\boldsymbol{1}_N^T\boldsymbol{x}_N=0$ for all $\boldsymbol{x}_N$ not proportional to $\boldsymbol{1}_N$.
\end{proof}

Note that $G_1$ cannot be put in the block diagonal form given and, in addition, it does not have any eigenvectors that are not proportional to $\boldsymbol{1}_1$. This is why the lemma only applies to $N>1$ and cannot be used to construct the second eigenvector of $G_2$.
\begin{corollary}
    For $N>1$, there must exist an eigenvector of $G_N$, $\boldsymbol{x}_N$, that satisfies $\boldsymbol{1}_N^T\boldsymbol{x}_N=0$.
\end{corollary}
\begin{proof}
    There exists a vector, $\boldsymbol{x}_{N+1}$ for $N>1$, defined to be an eigenvector, $\boldsymbol{x}_N$, of $G_N$ with an additional zero row. Since that vector is an eigenvector of $G_{N+1}$ and $\boldsymbol{1}_N^T\boldsymbol{x}_N$=0, by the theorem just proved, it immediately follows that $\boldsymbol{1}_{N+1}^T\boldsymbol{x}_{N+1}=0$ for $N\ge1$. By induction there must always be at least one eigenvector with that satisfies $\boldsymbol{1}_N^T\boldsymbol{x}_N=0$ when $N>1$.
\end{proof}
\begin{corollary}\corlab{iterate}
    If $\boldsymbol{x}_m$ is an eigenvector of $G_M$, for $m\in[2,M]$ and $M>1$, with $\boldsymbol{1}^T\boldsymbol{x}_m=0$, i.e.\ requiring that it not be equal to $\boldsymbol{1}_M$, then the vector with block form
    \begin{equation}
        \begin{pmatrix}
            \boldsymbol{x}_m\\\boldsymbol{0}_{N-M}
        \end{pmatrix}
    \end{equation}
    is an eigenvector of $G_N$ for all $N>M$.    
\end{corollary}
\begin{proof}
    This follows by induction from \lemref{iterate}.
\end{proof}
\begin{lemma}\lemlab{botrow}
    For $N>1$, $G_N$ has an eigenvector proportional to $\boldsymbol{x}_N$ where all the elements are $-1$ apart from the final one which is equal to $N-1$. The eigenvalue associated with this eigenvector, $\lambda_N$, has the value $1-\rho$.
\end{lemma}
\begin{proof}
From the definition of the matrix product, $G_N\boldsymbol{x}_N$ is
\begin{equation}
\begin{pmatrix}
1&\rho&\cdots&\rho\\
\rho&1&\cdots&\rho\\
\vdots&&\ddots&\vdots\\
\rho&\rho&\cdots&1
\end{pmatrix}
\begin{pmatrix}
    -1\\-1\\\vdots\\N-1
\end{pmatrix}=
\begin{pmatrix}
-(1-\rho)\\-(1-\rho)\\\vdots\\(1-\rho)(N-1)
\end{pmatrix}=
(1-\rho)
\begin{pmatrix}
    -1\\-1\\\vdots\\N-1
\end{pmatrix}.
\end{equation}
\end{proof}

These lemmas require that the eigenvectors of $G_2$ are
proportional to
\begin{equation}
    \begin{pmatrix}
        1\\1
    \end{pmatrix}\;\mathrm{and}\;
    \begin{pmatrix}
        -1\\\phantom{-}1
    \end{pmatrix}.
\end{equation}
Since $G_2$ is a real symmetric matrix of dimension $2$ these two vectors represent the complete set of eigenvectors for this problem size\cite[sec. 6B]{axler2015linear}. They may be normalized to choice.
\begin{definition}
    The projection matrix, $P_N$ for $N\ge1$, is an integer valued square matrix of dimension $N\ge1$ defined as follows:
    \begin{enumerate}[label=(\roman*)]
        \item all elements of the top row are $1$;
        \item all elements of the lower triangle are $-1$;
        \item all elements in the the upper triangle are $0$; and,
        \item for row $i>1$, the diagonal element $[P_N]_{ii}$ is equal to the number of negative elements in the same row, which is $i-1$.
    \end{enumerate}
\end{definition}
\begin{theorem}\thelab{eigen}
    The eigenvectors of $G_N$ are proportional to the vectors formed from the rows of $P_N$.
\end{theorem}
\begin{proof}
    For $N=1$ the theorem is trivially true as both $G_1$ and $P_1$ are a square matrix of dimension $1$ with element $1$ by definition. 
    
    For $N>1$, from \lemref{toprow}, there is always an eigenvector, $\boldsymbol{x}_1$, of $G_N$ proportional to $\boldsymbol{1}_N$ and, by definition, the top row of $P_N$ is equal to $\boldsymbol{1}_N$. Therefore the top row of $P_N$ is always proportional to an eigenvector of $G_N$. Also, from \lemref{botrow}, there is always an eigenvector of $G_N$, proportional to $\boldsymbol{x}_N$ as defined there. This is equal to the bottom row of $P_N$ by its definition. Therefore the bottom row of $P_N$ is always proportional to an eigenvector of $G_N$. For $N=2$ there are no other rows to consider and the theorem is true.

    If the theorem is true for $N\ge2$ then the eigenvectors of $G_N$ are proportional to the rows of $P_N$. From the definition of $P_N$ the first $N$ rows of $P_{N+1}$ are equal to the rows of $P_N$ when the first row is augmented by an additional $1$ in the final column and the remaining $N-1$ rows are augmented by an additional $0$ in the final column. By \lemref{toprow} and \corref{iterate} there are $N$ eigenvectors of $G_{N+1}$ that are proportional to these rows and, by \lemref{botrow}, the remaining\footnote{From the standard theorems for the Eigenvalue Problem, a matrix of dimension $N$ has exactly $N$ eigenvectors\cite[ch. 4]{arfken1999mathematical}.} $(N+1)$th.\ eigenvector of $G_{N+1}$ must be constructed in a manner that makes it proportional to the $(N+1)$th.\ row of $P_N$. Therefore, the theorem is true for $G_{N+1}$ and $P_{N+1}$ if it is true for $G_N$ and $P_N$. Since the theorem is true for both $N=1$ and $N=2$ it is therefore true for all $N$ by induction.
\end{proof}

\Theref{eigen} establishes that the eigenvectors of $G_N$ are proportional to the rows of $P_N$ for all $N\ge1$. To complete the eigendecomposition of $G_N$ it also is necessary to construct a matrix from $P_N$ that represents an orthogonal basis in $\mathbb{R}^N$. Due to the particular properties of $P_N$ this is relatively straightforward.
\begin{lemma}\lemlab{Pdiag}
    The matrix $P_N^TP_N$ is diagonal.
\end{lemma}
\begin{proof}
    Let $\boldsymbol{p}_i$ for $i\in[1,N]$ represent the vector formed from the $i$th row of $P_N$ for $N>1$. $P_N^TP_N$ is the matrix formed from all of the possible inner products of pairs of members of the set $\{\boldsymbol{p}_i\}$. From the definition of $P_N$ it immediately follows that $\boldsymbol{p}_i^T\boldsymbol{p}_i>0$ for all $i$. Consider $\boldsymbol{p}_1^T\boldsymbol{p}_j$ for $j\in[2,N]$: as all of the elements of $\boldsymbol{p}_1$ are $1$, this is equal to the sum of the elements of $\boldsymbol{p}_j$ which is zero by construction. Consider $\boldsymbol{p}_i^T\boldsymbol{p}_j$ for $i>1$ and $i<j$: from the definition of $P_N$ this is equal to $-\boldsymbol{p}_i^T\boldsymbol{1}_i$ which is equal to $-1$ times the sum of the elements of $\boldsymbol{p}_i$ and so is also zero by construction. From the definition of the inner product $\boldsymbol{p}_i^T\boldsymbol{p}_j=\boldsymbol{p}_j^T\boldsymbol{p}_i$ for all $i,j$. Thus $P_N^TP_N$ is diagonal for $N>1$ and $P_1^TP_1$ is diagonal trivially. Therefore $P_N^TP_N$ is diagonal always.
\end{proof}
\begin{definition}
    For $N>1$, let $A_N$ be a integer valued diagonal matrix of dimension $N-1$ with the values $j(j-1)$ along the diagonal for $j\in[2,N]$.
\end{definition}
\begin{definition}
    Let $B_N$ be the diagonal matrix of dimension $N$ with the block-matrix form
    \begin{equation}
    B_N=\begin{pmatrix}
        N&\boldsymbol{0}_{N-1}^T\\
        \boldsymbol{0}_{N-1}&A_N
    \end{pmatrix}.       
    \end{equation}
\end{definition}
\begin{definition}
    Let the matrix $Q_N$ be a real valued square matrix of dimension $N$ given by\footnote{Since $B_N$ is diagonal the power $B_N^r$ may be defined to apply elementwise along the diagonal, for all $r\in\mathbb{R}$. Thus the inverse of $B_N$ is the matrix with the diagonal elements replaced by their reciprocals and the square root of $B_N$ is the matrix with diagonal elements replaced by their square roots.}
    \begin{equation}
        Q_N=B_N^{-1/2}P_N.
    \end{equation}
\end{definition}
\begin{theorem}\thelab{Qorth}
    The matrix $Q_N$ is orthogonal.
\end{theorem}
\begin{proof}
    It has already been shown that $P_N^TP_N$ is diagonal. Therefore it is sufficient to show that the elements along the diagonal of that matrix equal the elements along the diagonal of $B_N$. 
    
    From the definition of $P_N$, $\boldsymbol{p}_1^T\boldsymbol{p}_1=N$. Consider $\boldsymbol{p}_j$ for $j\in[2,N]$: the sum of the squares of the elements of this vector equals $(j-1)+(j-1)^2=j(j-1)$. Thus the values $\{\boldsymbol{p}_j^T\boldsymbol{p}_j\}$ are the diagonal elements of $A_N$ from the definition above. Therefore the diagonal elements of $P_N^TP_N$ are equal to the diagonal elements of $B_N$ and it follows that $Q_N^TQ_N=I_N$, where $I_N$ is the identity matrix of dimension $N$, and so $Q_N$ is an orthogonal matrix.
\end{proof}

\Theref{Qorth} shows that $Q_N$ is an orthogonal matrix and, since it is directly proportional to $P_N$, it represents an eigenmatrix of $G_N$. The theorems above also establish that all of the eigenvalues of $G_N$ are $1-\rho$ apart from the one associated with the top row of $Q_N$ which is $1+(N-1)\rho$. This is the complete eigendecomposition of $G_N$ into the form given in \secref{eigen}. 
\begin{corollary}\corlab{inv}
    The inverse of $G_N$ is $P_N^TD^{-1}B_NP_N$, where $G_N$, $P_N$ and $B_N$ are as given above and $D_N$ is a diagonal matrix with the eigenvalue corresponding to the top row of $P_N$, or $1+(N-1)\rho$, as the first element and $1-\rho$ for the others.
\end{corollary}
\begin{proof}
    Since $Q_N$ is an orthogonal matrix formed from the eigenvectors of $G_N$ and $D_N$ is the associated matrix of eigenvalues, it follows that $G_N=Q_N^TD_NQ_N$. The expression is derived by applying the definition of $Q_N$ above and rules of matrix arithmetic.
\end{proof}
\begin{definition}
    The matrix $H_N$ is a real valued symmetric matrix of dimension $N>1$ with $-1-(N-2)\rho$ along the diagonal and $\rho$ everywhere else.
\end{definition}
\begin{theorem}\thelab{inv2}
    For $N>1$, the inverse of $G_N$ is $H_N/f_N$ where 
    \begin{equation}
    f_N=(N-1)\rho^2-(N-2)\rho-1.
    \end{equation}
\end{theorem}
\begin{proof}
    For $N>1$ the rows of $G_N$ may be written as the transpose of the vectors $\boldsymbol{g}_i=(1-\rho)\boldsymbol{e}_i+\rho\boldsymbol{1}_N$ for $i\in[1,N]$, where $\boldsymbol{e}_i$ is the Euclidean basis vector in $\mathbb{R}^N$ with $1$ in the $i$th.\ row and $0$ everywhere else. Similarly the rows of $H_N$ may be written as the transpose of the vectors $\boldsymbol{h}_i=-\{1+(N-1)\rho\}\boldsymbol{e}_i+\rho\boldsymbol{1}_N$. Therefore the $(i,j)$ element of the matrix product $G_NH_N$ is
    \begin{align}
    \boldsymbol{g}_i^T\boldsymbol{h}_j=
    &\,\{(1-\rho)\boldsymbol{e}_i+\rho\boldsymbol{1}_N\}^T
    [-\{1+(N-1)\rho\}\boldsymbol{e}_j+\rho\boldsymbol{1}_N]\\
    =&-(1-\rho)\{1+(N-1)\rho\}\boldsymbol{e}_i^T\boldsymbol{e}_j
    +(1-\rho)\rho\boldsymbol{e}_i^T\boldsymbol{1}_N\\
    &-\{1+(N-1)\rho\}\rho\boldsymbol{1}_N^T\boldsymbol{e}_j
    +\rho^2\boldsymbol{1}_N^T\boldsymbol{1}_N\nonumber\\
    =&\,f_N\boldsymbol{e}_i^T\boldsymbol{e}_j
    \end{align}
    since $\boldsymbol{e}_i^T\boldsymbol{1}_N=1$ and $\boldsymbol{1}_N^T\boldsymbol{1}_N=N$. As $\{\boldsymbol{e}_i\}$ forms an orthogonal basis in $\mathbb{R}^N$, it immediately follows that
    \begin{equation}
        G_NH_N=f_NI_N\;\Rightarrow\;H_N/f_N=G_N^{-1}.
    \end{equation}
\end{proof}
\noindent For $N=1$ the inverse of $G_N$ is trivial.
\begin{corollary}\corlab{singular}
    For $N>1$, $G_N$ is singular if $\rho=1$ or $\rho=1/(1-N)$.
\end{corollary}
\begin{proof}
    These are the roots of the quadratic equation $f_N=0$.
\end{proof}
\section{Author's Statement on the Availability of Data and Code to Execute the Analysis Presented Herein}\seclab{data}
Empirical work presented here is executed in \textit{Python} code\cite{rossum1995python} using the standard ``open source'' toolkit (\textit{Pandas}\cite{mckinney2010data}, \textit{Numpy}\cite{harris2020array}, \textit{SciPy}\cite{scipy2020}) as found on Google's \textit{Colab} system\cite{googlecolab}. Analytical notebooks are archived on the author's personal \textit{GitHub} repository\cite{gillergithub} and the notebook \texttt{\smaller Index\_Pairwise\_Correlations.ipynb}, which may be found in the folder \texttt{\smaller Financial-Data-Science-in-Python}, is used for the analysis presented herein. This code base is under development, but the version control system presented by the \textit{GitHub} website permits the specific version in use to generate the figures and tables incorporated in this document to be extracted by users. Data on stocks is extracted programatically via the \texttt{\smaller yfinance} package\cite{yfinance} from sources made available to the general public by \textit{Yahoo! Inc.} and from \textit{Wikipedia}\cite{wiki2024spx}.
%
% end
%
\bibliographystyle{plain}
\bibliography{citation}
\end{document}